\documentclass[11pt, a4paper]{article}
\usepackage{fullpage}
\usepackage[utf8]{inputenc}
\usepackage{amsmath}
\usepackage{amssymb,amsthm,graphicx}
\usepackage[dvips]{epsfig}
\usepackage{amsfonts}
\usepackage{xcolor}
\usepackage{hyperref}
\usepackage{enumerate}
\usepackage{thm-restate}
\usepackage{xspace}
\usepackage{complexity}

\usepackage{multirow}

\newtheorem{theorem}{Theorem}

\newtheorem{corollary}[theorem]{Corollary}

\newtheorem{lemma}[theorem]{Lemma}

\theoremstyle{remark}
\newtheorem{remark}[theorem]{Remark}

\theoremstyle{definition}

\newcommand{\N}{\mathbb{N}}
\newcommand{\Z}{\mathbb{Z}}

\newcommand{\F}{\mathbb{F}} 

\definecolor{darkgreen}{rgb}{0,.5,0}

\newcommand{\mc}{\multicolumn{1}{c}}

\newcounter{sideremark}

\newcommand{\probVEST}{\textsc{VEST}\xspace}
\newcommand{\probkSum}{\textsc{$k$-Sum}\xspace}

\newcommand{\probAtmkSumRep}{\textsc{At-Most-$k$-Sum with Repetitions and Target 1}\xspace}
\newcommand{\probUHS}{\textsc{Unique Hitting Set}\xspace}
\newcommand{\probkEC}{\textsc{$k$-Exact Cover}\xspace}
\newcommand{\probkProdRep}{\textsc{$k$-Product with Repetitions}\xspace} 
\newcommand{\probMkProdRep}{\textsc{Matrix $k$-Product with Repetitions}\xspace} 
\newcommand{\probMkProdZero}{\textsc{Matrix $k$-Product with Repetitions resulting to Zero Matrix}\xspace}
\newcommand{\probMkProdId}{\textsc{Matrix $k$-Product with Repetitions resulting to Identity Matrix}\xspace}

\usepackage{boxedminipage}
\usepackage{tabularx}
\newcommand{\pbDef}[3]{
\noindent
\begin{center}
\begin{boxedminipage}{1 \columnwidth}
\textsc{#1}\\[5pt]
\begin{tabularx}{\textwidth}{l X}
Input: & #2\\
Question: & #3
\end{tabularx}
\end{boxedminipage}
\end{center}
}

\title{Counting Vanishing Matrix-Vector Products\thanks{
C.B. and V.K. were supported by Austrian Science Fund (FWF, project Y1329),
K.S. was supported by DFG~Research~Group ADYN via grant DFG 411362735,
M.S. acknowledges support by the project ``Grant Schemes at CU'' (reg. no. CZ.02.2.69/0.0/0.0/19\_073/0016935) and GA\v{C}R grant 22-19073S.
}} 

\author{
Cornelius Brand
\\ \small{Algorithms and Complexity Group, TU Wien, Austria}
\and
Viktoriia Korchemna
\\ \small{Algorithms and Complexity Group, TU Wien, Austria}
\and
Michael Skotnica
\\ \small{Department of Applied Mathematics, Charles University, Prague, Czech Republic}
\and
Kirill Simonov
\\ \small{Hasso Plattner Institute, University of Potsdam, Germany}}

\date{}

\begin{document}

\maketitle

\begin{abstract}
Consider the following parameterized counting variation of the classic subset sum problem,
which arises notably in the context of higher homotopy groups of topological spaces:
Let $\mathbf{v} \in \mathbb{Q}^d$ be a rational vector, 
$(T_{1}, T_{2} \ldots T_{m})$ a list of $d \times d$ rational matrices,
$S \in \mathbb{Q}^{h \times d}$ a rational matrix
not necessarily square and $k$ a parameter. The goal is to compute
the number of ways one can choose $k$ matrices $T_{i_1}, T_{i_2}, \ldots, T_{i_k}$ from the list
such that $ST_{i_k} \cdots T_{i_1}\mathbf{v} = \mathbf{0} \in \mathbb{Q}^h$.

In this paper, we show that this problem is $\# \W[2]$-hard for parameter $k$.
As a consequence, computing the $k$-th homotopy group of a $d$-dimensional 1-connected topological space
for $d > 3$ is $\# \W[2]$-hard for parameter $k$.
We also discuss a decision version of the problem and its several modifications for which we show $\W[1]/\W[2]$-hardness. This is in contrast to the parameterized $k$-sum problem, which is only $\W[1]$-hard (Abboud-Lewi-Williams, ESA'14). 
In addition, we show that the decision version of the problem without parameter is an undecidable problem,
and we give a fixed-parameter tractable algorithm for matrices of bounded size over finite fields, parameterized the matrix dimensions and the order of the field.
\end{abstract}
\section{Introduction}
Topology is one of the most important and active areas of mathematics, emerging from vast generalizations of geometry (see, e.g.,~\cite{flegg2001geometry} for a gentle introduction along this path).
In full generality, it studies fundamental properties of \emph{topological spaces},
which generalize a broad array of geometric objects (including manifolds, Hilbert spaces, algebraic varieties and even embeddings of graphs).
The concept of a topological space allows to speak in a very general manner about the ``shape'' of a space,
and a prime goal of topology consists in classifying spaces according to their shapes.
For instance, it is intuitively obvious that a mug with a handle and a football should belong to distinct classes of shapes, for instance because one has a hole in it and the other, preferably, does not. Whether or not, then, a mug with sharp edges and a doughnut should belong to the same class is a different question, and good reasons exist for choosing either way of answering it.

Thus, clearly, any such classification depends on the precise way in which the classes are defined and the structures provided on top of purely topological information (such as differential information, i.e., about ``sharp edges''); one particularly important way of doing so is to make a single class out of all those shapes that can be deformed into each other according to specific rules retaining.
The usual notion of equivalence under deformation of shapes corresponding to general topological spaces is furnished by \emph{homotopy},
which, very roughly speaking, identifies any two shapes that can be obtained from one another through arbitrary deformations without ``tearing'' or ``cutting'' (and hence identifying the mug with the doughnut, while differentiating both from the football).
 
Associated to this notion are the so-called \emph{homotopy groups} of a topological space,
denoted $\pi_k$, for $k \geq 1.$ 
The most intuitive of them is the group $\pi_1$, which is often called the \emph{fundamental group} of the space. 
It captures certain data about the different ways that loops (that is, closed curves in the space) can pass through the space.
The higher homotopy groups ($k > 1$) correspond to ways of routing higher-dimensional ``loops'' in the space,
and Whitehead's Theorem provides a crucial equivalence between the structure of homotopy groups and the homotopy class of a broad category of topological spaces called CW-complexes \cite{whitehead1949combinatorial,whitehead1949combinatorial2}.
The present paper deals with an intermediate problem related to the computation of homotopy groups, which allows to show lower bounds for the complexity of computing the higher homotopy groups of a topological space.

Before speaking about computational tasks associated with topological spaces,
one needs to define how a topological space is even represented.
While the generality of the concept may make it seem hard to come up with such a representation in general, the usual path taken in computational topology is as follows:
Many topological spaces can be described by finite structures, e.g., by abstract simplicial complexes, which are simply collections of point sets closed under taking subsets, and it hence suffices to provide the maximal subsets of a simplicial complex to specify it in full.
Such structure can then be used as an input for a computer and therefore,
it is natural to ask how hard it is to compute these homotopy groups of a given topological space, represented by an abstract simplicial complex.

Novikov in 1955~\cite{Novikov55} and independently Boone in 1959~\cite{Boone59}
showed undecidability of the word problem for groups.
Their result also implies undecidability of computing the fundamental group.
In fact, even determining whether the fundamental group of a given topological space is trivial
is undecidable.

On the other hand, for 1-connected spaces (for those, whose $\pi_1$ is trivial) it is known that their $\pi_k$ for $k > 1$
are finitely generated abelians group which are always isomorphic to groups of the form
$  \Z^n \oplus \Z_{p_1} \oplus \Z_{p_2} \oplus \cdots \oplus\Z_{p_m},$
where $p_1, \ldots, p_m$ are powers of prime numbers.%
\footnote{Note that $\Z^n$ is a direct sum of $n$ copies of $\Z$ while $\Z_{p_i}$ is a finite cyclic group of order $p_i$.}
An algorithm for computing $\pi_k$ of a 1-connected space, where $k > 1$, was first introduced by Brown in 1957~\cite{Brown57}.

In 1989, Anick~\cite{Anick89} proved that computing the rank of $\pi_k$,
that is, the number of direct summands isomorphic to $\Z$ (represented by $n$ in the expression above)
is $\#\P$-hard for 4-dimensional 1-connected spaces.\footnote{When $k$ is a part of the input and represented in unary.}
Another computational problem called \probVEST, which we define below, was used in Anick's proof as an intermediate step.
Briefly said, $\#\P$-hardness of \probVEST implies $\#\P$-hardness of computing the rank of $\pi_k$,
which is the motivation for studying the problem in the present article.

\paragraph*{Vector Evaluated After a Sequence of Transformations (\probVEST).}
The input of this problem defined by Anick~\cite{Anick89}
is a vector $\textbf{v} \in \mathbb{Q}^d$, a list $(T_1,\ldots, T_m)$
of rational $d \times d$ matrices and a rational matrix $S \in \mathbb{Q}^{h \times d}$
where $d,m,h \in \N$.

For an instance of \probVEST let an \emph{$M$-sequence} be a sequence of integers $M_1, M_2, M_3, \ldots$, where
\begin{align*}
  M_k := |\{(i_1,\ldots, i_k) \in \{1, \ldots, m\}^k; ST_{i_k}\cdots T_{i_1} \mathbf{v} = \mathbf{0}\}|.
\end{align*}
Given an instance of \probVEST and $k \in \N$, the goal is to compute $M_k$.

From an instance of \probVEST, it is possible to construct a corresponding algebraic structure called \emph{$123H$-algebra}
in polynomial time whose \emph{Tor-sequence} is equal to the $M$-sequence of the original instance of a \probVEST.
This is stated in \cite[Theorem 3.4]{Anick89} and it follows from \cite[Theorem 1.3]{Anick85} and \cite[Theorem 7.6]{Anick87}.

Given a presentation of a $123H$-algebra, one can construct a corresponding 4-dimensional simplicial complex in polynomial time
whose sequence of ranks $(\textup{rk }\pi_2, \textup{rk }\pi_3, \ldots)$ is related to the Tor-sequence of the $123H$-algebra.
In particular, it is possible to compute that Tor-sequence from the sequence of ranks using an $\FPT$ algorithm. (To be defined in the next paragraph). This follows from \cite{Roos79} and \cite{Cadek14_2}. To sum up, hardness of computing $M_k$ of \probVEST implies hardness of computing $\pi_k$.

\paragraph*{Parameterized Complexity and the $\W$-hierarchy}
Parameterized complexity classifies decision or counting computational problems with respect to a given parameter(s).
For instance, one can ask if there exists an independent set of size $k$ in a given graph
or how many independent sets of size $k$ (for counting version) are in a given graph, respectively,
where $k$ is the parameter. From this viewpoint, we can divide problems into several groups which
form the \emph{$\W$-hierarchy}.
\begin{align*}
    \FPT \subseteq \W[1] \subseteq \W[2] \subseteq \cdots \subseteq \XP
\end{align*}

The class $\FPT$ consists of decision problems solvable in time $f(k)n^{c}$,
where $f(k)$ is a computable function of the parameter $k$, $n$ is the size of input and $c$ is a constant,
while the class $\XP$ consists of decision problems solvable in time $cn^{f(k)}$.
The class $\W[1]$ consists of all problems
which admit a parameterized reduction to the satisfiability problem of a boolean circuit of constant depth
with AND, OR  and NOT gates such that there is at most 1 gate of higher input size than 2
on each path from the input gate to the final output gate (this number of larger gates is called \emph{weft}), where the parameter is the number of input gates set to TRUE. Here, a parameterized reduction from a parameterized problem $A$ to a parameterized problem $B$ is an algorithm that, given an instance $(x, k)$ of $A$, in time $f(k) n^c$ produces an equivalent instance $(x, k')$ of $B$ such that $k' \le g(k)$, for some computable functions $f(\cdot)$, $g(\cdot)$, and a constant $c$.
See Figure~\ref{Fig:clique} for an example of a reduction showing $\W[1]$-completeness of finding independent set of size $k$.

\begin{figure}
  \centering
  \includegraphics[scale=.862]{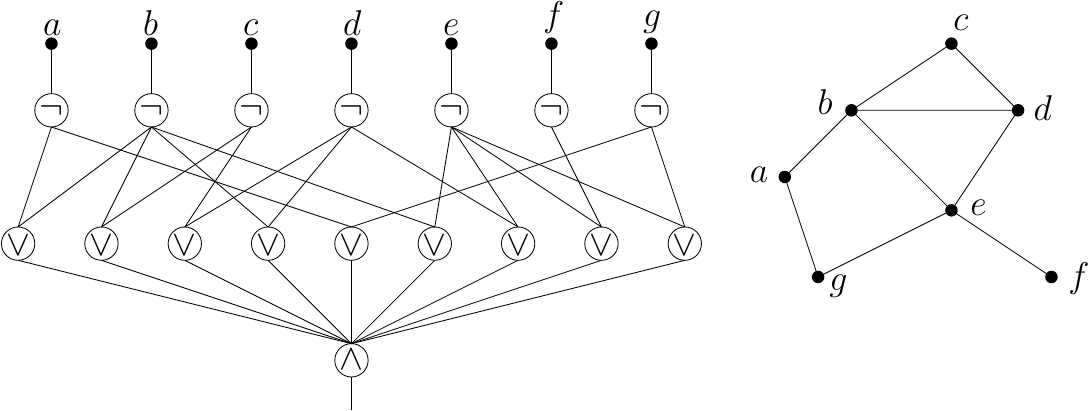}
  \caption{A boolean circuit solving the problem of existence of an independent set of size $k$ in the graph on the left.
  There is an independent set of size $k$ in the graph if and only if the boolean circuit outputs TRUE
  for an input consisting of exactly $k$ true values.
  }
  \label{Fig:clique}
\end{figure}

The class $\W[i]$ then consists of problems that admit a parameterized reduction to the satisfiablity problem of a boolean circuit of a constant depth and weft at most $i$, parameterized by the number of input gates set to TRUE.

It is only known that $\FPT \subsetneq \XP$, while the other inclusions in the \W-hierarchy are not known to be strict.
However, it is strongly believed that $\FPT \subsetneq \W[1]$.
\emph{Therefore, one cannot expect existence of an algorithm solving a $\W[1]$-hard problem in time $f(k)n^{c}$
where $f(k)$ is a computable function of $k$ and $c$ is a constant.}
For the detailed presentation of \W-hierarchy and parameterized complexity in general we refer the reader to~\cite{Flum04}.

\smallskip
Analogously, one can define classes $\FPT$ and $\XP$ for counting problems.
That is, a class of counting problems solvable in time $f(k)n^{c}$ or $cn^{f(k)}$, respectively.
Problems for which there is a parameterized counting reduction to a problem of counting solutions for a boolean circuit of
constant depth and weft at most $i$ then form class $\# \W[i]$.
Note that there are decision problems from $\FPT$ whose counting versions are $\# \W[1]$-hard,
e.g., counting paths or cycles of length $k$ parameterized by $k$~\cite{Flum04_counting}.
Similarly to the decision case, \emph{if a counting problem is shown to be $\# \W[i]$-hard for some $i$
one should not expect existence of an algorithm solving this problem in time $f(k)n^{c}$.}
For more details on parameterized counting we refer the reader to~\cite{Flum04_counting}.

\smallskip
In our case, the number $k$ of the homotopy group $\pi_k$ plays the role of the parameter.
In 2014 \v{C}adek et al.~\cite{Cadek14} proved that computing $\pi_k$
(and thus, also computing the rank of $\pi_k$) is in $\XP$ parameterized by $k$.

A lower bound for the complexity from the parameterized viewpoint was obtained by Matou\v{s}ek in 2013~\cite{Matousek13}.
He proved that computing $M_k$ of a \probVEST instance is $\# \W[1]$-hard. This also implies $\# \W[1]$-hardness for the original problem of computing the rank of higher homotopy groups $\pi_k$ (for 4-dimensional 1-connected spaces) for parameter $k$.
Matou\v{s}ek's proof also works as a proof for $\#\P$-hardness and it is shorter and considerably easier than the original proof of Anick in \cite{Anick89}.

In this paper, we strengthen the result of Matou\v{s}ek and show that computing $M_k$ of a \probVEST instance is $\# \W[2]$-hard.
Our proof is even simpler than the previous proof of $\# \W[1]$-hardness.

\begin{theorem}\label{thm:w2hardness}
  Given a \probVEST instance, computing $M_k$ is $\# \W[2]$-hard when parameterized by $k$.
\end{theorem}

Theorem~\ref{thm:w2hardness} together with the result of Anick~\cite{Anick89} implies the following.

\begin{corollary}
  Computing the rank of the $k$-th homotopy group of a $d$-dimensional 1-connected space
  for $d > 3$ is $\# \W[2]$-hard for parameter $k$.
\end{corollary}


\begin{remark}
Note that computing $M_k$ of a \probVEST instance is an interesting natural self-contained problem
even without the topological motivation.
We point out that our reduction showing $\# \W[2]$-hardness of this problem
uses only 0, 1 values in the matrices and the initial vector $\mathbf{v}$.
Moreover, each matrix will have at most one 1 in each row.
Therefore, such construction also shows $\# \W[2]$-hardness of computing $M_k$ of a \probVEST instance in the $\Z_2$ setting.
That is, for the case when $T_1, T_2, \ldots T_m \in \Z_2^{d \times d}, S \in \Z_2^{h \times d}$ and $\mathbf{v} \in \Z_2^d$.
\end{remark}

\paragraph*{The Decision Version of \probVEST}
We also provide a comprehensive overview of the parameterized complexity of \probVEST as a decision problem, where given an instance of \probVEST one needs to determine whether $M_k > 0$. In addition to the standard variant of the problem, we consider several modifications of \probVEST: when the matrices have constant size, when the matrix $S$ is the identity matrix, when we omit the initial vector and the target is identity/zero matrix etc.

Unfortunately, even considering the simplifications above, we show that nearly all versions in our consideration are $\W[1]$- or $\W[2]$-hard.
The following table is an overview of our results.
\begin{center}
\begin{tabular}{ |l||c|c|l|c|c| }
\hline
Size of matrices & a) $\mathbf{v}$ and $S$ & b) only $\mathbf{v}$ & & c) only $S$ & d) no $\mathbf{v}$, no $S$ \\ \hline\hline
\multirow{2}{*}{1. $1\times 1$} & \multirow{2}{*}{$\P$} & \multirow{2}{*}{$\P$} & $\mathbf{0}$ & $\P$ & $\P$ \\
 & & & $I$ & $\W[1]$-hard & $\W[1]$-hard \\ \hline
\multirow{2}{*}{2. $2\times 2$} & \multirow{2}{*}{$\W[1]$-hard} & \multirow{2}{*}{$\W[1]$-hard} & $\mathbf{0}$ & \multirow{2}{*}{$\W[1]$-hard} & \multirow{2}{*}{$\W[1]$-hard} \\
 & & & $I$ & & \\ \hline
 \multirow{2}{*}{3. input size} & \multirow{2}{*}{$\W[2]$-hard} & \multirow{2}{*}{$\W[2]$-hard} & $\mathbf{0}$ & $\W[2]$-hard & $\W[2]$-hard \\
 & & & $I$ & $\W[1]$-hard & $\W[1]$-hard  \\ \hline
\end{tabular}
\end{center}
The first column stands for the standard \probVEST while the second stands for the \probVEST without the special matrix $S$ or alternatively,
for the case when $S$ is the identity matrix. Therefore, the hardness results for the first column follow from the second.

The third and the fourth columns are without the initial vector $\mathbf{v}$. In this case, it is natural to assume the following 
two targets for the result of the sought matrix product: the zero matrix (the rows labeled by $\mathbf{0}$) and the identity matrix (the rows labeled by $I$).
Again, the hardness results for the third column follow from the fourth.

Regarding the $1 \times 1$ case, the only nontrivial case is when the target is $I = 1$.
The $\W[1]$-hardness results for the $1 \times 1$ case also implies $\W[1]$-hardness
for the $2 \times 2$ case and the input size case when the target is the identity matrix.

Therefore, in Section~\ref{s:variants_of_VEST} we prove hardness for
\begin{itemize}
	\item ``1 d) $I$'' (Theorem~\ref{thm:k_product_hard}),
	\item ``2 b)'' (Theorem~\ref{thm:k_product_2x2_with_v_hard}),
	\item ``2 d) $\mathbf{0}$'' (Theorem~\ref{kProd0MatrHard}).
\end{itemize}
The $\# \W[2]$-hardness for ``3 c)'' follows from the proof of Theorem~\ref{thm:w2hardness} (see Remark~\ref{r:decision_vest_hard}) and we show that
``3 b)'' and ``3 d) $\mathbf{0}$'' 
are equivalent to ``3 a)'' under parameterized reduction (Theorems~\ref{thm:without_S_equivalence}, \ref{thm:equivalence_matrix_k_product_vest}).

\paragraph*{Fixed-Parameter Tractability over Finite Fields}
Reductions from the previous section show that \probVEST remains hard even on highly restricted instances, such as binary matrices with all the ones located along the main diagonal, or matrices of a constant size. However, it turns out that combination of this two restrictions -- on the field size and the matrix sizes -- makes even the counting version of \probVEST tractable. 

We proceed by lifting tractability to the matrices of unbounded size but with all non-zero entries occurring in at most the $p$ first rows. 
\begin{restatable}{theorem}{thmFPT}
\label{thm:FPT_result}
Given an instance of \probVEST and $k\in \N$, computing $M_k$ is $\FPT$ when parameterized by $|\F|$ and $p$, if all but the first $p$ rows of the input matrices are zeros.
\end{restatable}
The problem remains $\FPT$ with respect to $|\F|$ and $p$ even if the task is to find the minimal $k$ for which the vanishing sequence of length $k$ exists, or to report that there is no such $k$. 

\paragraph*{Undecidability of \probVEST Without Parameter} \label{s:W[2]_hardness}
In contrast, we show in the last section (Section~\ref{s:undecidability})
that for $\F=\mathbb{Q}$ the problem of determining whether there exists $k$ such that $M_k > 0$ for an instance of \probVEST is an undecidable problem
(even for the case where $T_1, \ldots, T_m$ are of size $4 \times 4$).

\section{The Proof of \#\W[2]-hardness of \probVEST} 
In this section, we prove that computing $M_k$ of a \probVEST is $\# \W[2]$-hard (Theorem~\ref{thm:w2hardness}).
Our reduction is from the problem of counting dominating sets of size $k$
which is known to be $\# \W[2]$-complete (see~\cite{Flum04_counting}) and which we recall in the paragraph below.

For a graph $G(V,E)$ and its vertex $v \in V$
let $N[v]$ denote the closed neighborhood of a vertex $v$.
That is, $N[v] := \{u \in V; \{u,v\} \in E\} \cup \{v\}$.
A \emph{dominating set} of a graph $G(V,E)$ is a set $U \subseteq V$
such that for each $v$ there is $u \in U$ such that $v \in N[u]$.
\pbDef{Number of dominating sets of size $k$}
{A graph $G(V,E)$ and a parameter $k$.}
{How many dominating sets of size $k$ are in $G$?}

\begin{proof}[Proof of Theorem~\ref{thm:w2hardness}]
  As we said, we show an FPT counting reduction from the problem of counting dominating sets of size $k$ to \probVEST.
  
  Let $G=(V,E)$ be the input graph and let $n = |V|$.
  The corresponding instance of \probVEST will consist of $n$ matrices $\{T_u \: : \: u \in V\}$ of size $4n \times 4n$, one for each vertex,
  and matrix $S$ of the same size.
  Whence, the initial vector $\mathbf{v}$ must be of size $4n$.
  For each vertex $u \in V$, we introduce four new coordinates $u_1, \ldots, u_4$ and set $\mathbf{v}_{u_1} =1, \mathbf{v}_{u_2} = \mathbf{v}_{u_3} = 0$ and $\mathbf{v}_{u_4} = 1$.
  
  We define the matrices $\{T_u \: : \: u \in V\}$ and $S$ by describing their behavior.
  Let $\mathbf{x}$ be a vector which is going to be multiplied with a matrix $T_u$ (that is, some intermediate vector obtained from $\mathbf{v}$ after potential multiplications).
  The matrix $T_u$ sets $\mathbf{x}_{w_1}$ to zero for each $w \in N[u]$,
  which corresponds to domination of vertices in $N[u]$ by the vertex $u$,
  and also sets $\mathbf{x}_{u_2}$ to $\mathbf{x}_{u_3}$ and $\mathbf{x}_{u_3}$ to $\mathbf{x}_{u_4}$.
  The rest of the entries of $\mathbf{x}$ including $\mathbf{x}_{u_4}$ are kept, see Figure~\ref{fig:matrix_T_u}.
  \begin{figure}
  \centering
  $\begin{pmatrix}
  	0 & 0 & 0 & 0\\
    0 & 0 & 1 & 0\\
    0 & 0 & 0 & 1\\
    0 & 0 & 0 & 1
  \end{pmatrix}$
  \caption{The submatrix of $T_u$ consisting of rows and columns $u_1, \ldots, u_4$.
   The rest of the non-diagonal entries of $T_u$ are zeros. The diagonal entries $T_u^{w_1,w_1}$ for $w \in N[u]$ are zeros,
   the rest of the diagonal entries are ones.}\label{fig:matrix_T_u}
  \end{figure}  
   
  The matrix $S$ then nullifies coordinates $u_3, u_4$ and keeps the coordinates $u_1$ and $u_2$ for each $u \in V$.
  In other words, $S$ is diagonal such that $S^{u_1, u_1} = S^{u_2,u_2} = 1$ and $S^{u_3, u_3} = S^{u_4,u_4} = 0$.
  
  The parameter remains equal to $k$.
 
  For correctness, let $u^1, \ldots, u^k$ be any vertices from $V$, and let $\mathbf{r}$ be the vector obtained from $\mathbf{v}$ after multiplying by the matrices $T_{u^1}, \ldots, T_{u^k}$ (observe that the order of multiplication does not matter since all $T_u$, $u\in V$, pairwise commute).
  By construction, for every vertex $u \in V$, the entry $\mathbf{r}_{u_1} = 0$ if and only $u$ is dominated by some $u^i$, $i\in[k]$,
  and $\mathbf{r}_{u_2} = 0$ if and only if $T_u$ appears among $T_{u^1}, \ldots, T_{u^k}$ at most once.
  Indeed, if $T_u$ is selected once then $\mathbf{r}_{u_2} = \mathbf{v}_{u_3} = 0$
  while if it is selected more than once then $\mathbf{r}_{u_2} = \mathbf{v}_{u_4} = 1$.
  If $T_u$ is not among $T_{u^1}, \ldots, T_{u^k}$ then $\mathbf{r}_{u_2} = \mathbf{v}_{u_2} = 0$.

  Therefore, $\mathbf{r}=T_{u^1}\ldots T_{u^k} \mathbf{v}$ is a zero vector if and only if $u^1, \ldots ,u^k$ are pairwise distinct and form the dominating set in $G$. This provides a one-to-one correspondence between subsets of matrices yielding the solution of \probVEST and dominating sets of size $k$ in $G$. It remains to note that every such subset of matrices gives rise to $k!$ sequences that have to be counted in $M_k$. Hence, $M_k = k!D_k$ where $D_k$ is the number of dominating sets of size $k$ in $G$. The reduction is clearly $\FPT$
  since the construction does not use parameter $k$ and is polynomial in size of the input.
\end{proof}
%
\begin{remark}\label{r:decision_vest_hard}
Note that the decision version of the problem of \textsc{Dominating Sets of Size $k$} is $\W[2]$-hard.
For showing \W[2]-hardness of the decision version of \probVEST
we need not deal with the repetition of matrices.
In particular, we do not need the special coordinates $u_2, u_3, u_4$ and therefore,
the corresponding instance of \probVEST can consist only of diagonal 0,~1 matrices of size $n \times n$.
\end{remark}

\section{Modifications of \probVEST}\label{s:variants_of_VEST}

In this section, we prove hardness for the variants of the decision version of \probVEST we have discussed in the introduction.
First of all, we recall a well-known $\W[1]$-hard \probkSum problem. See also \cite{AbboudLW14}.
%
%

\pbDef{\probkSum}
{A set $A$ of integers and a parameter $k$.}
{Is it possible to choose $k$ distinct integers from $A$ such that their sum is equal to zero?}

We note that in the versions of \textsc{$k$-Sum} studied in the literature the goal is to pick \emph{distinct} elements of the input set in order to achieve 0
or eventually another number. However, the motivation for \probVEST, to the contrary, does not suggest that the matrices chosen for the product have to be distinct. Thus, in order to model \probVEST by \probkSum, it is more natural to also allow repetition of numbers. For our particular proofs, we will use the following version with target number 1.

\pbDef{\probAtmkSumRep}
{A set $A$ of integers and parameter $k$.}
{Is it possible to choose \emph{at most} $k$ integers from $A$ (possibly with repetition) such that their sum is equal to $1$?}

We are not aware of any previous studies on parameterized complexity of \probAtmkSumRep,
nor does it seem that there exists a simple parameterized reduction from the original variant of the problem to the one with repetitions.
Therefore, in Appendix (see~\ref{sec:appendix}, Theorem~\ref{thm:at_most_k_sum_with_repetition_hardness})
we prove $\W[1]$-hardness of this problem directly.
Our reduction is from the problem of \probkEC, which is known to be $\W[1]$-hard (see~\cite{Downey95}).

\pbDef{\probkEC}
{A universe $U$, a collection $\mathcal{C}$ of subsets of $U$ and a parameter $k$.}
{Can $U$ be partitioned into $k$ sets from $\mathcal{C}$?}

When we assume multiplication instead of addition the following problem arises.

\pbDef{\probkProdRep}
{A set $A$ of rational numbers and a parameter $k$.}
{Is it possible to choose $k$ numbers from $A$ (possibly with repetitions) such that their product is equal to 1?}

$\W[1]$-hardness for this problem might be a folklore result but we present a complete proof using
a reduction from \probkEC. 


\begin{theorem}\label{thm:k_product_hard}
  \probkProdRep is $\W[1]$-hard parameterized by $k$.
\end{theorem}
\begin{proof}
  We show a parameterized reduction from \probkEC. For each element $u \in U$ we associate one prime $p_u$,
  then for each $C \in \mathcal{C}$  we set $i_C := p \prod_{c \in C} p_c$ where $p$ is a prime which is not used for any element from $U$
  and $s := \frac{1}{p^k\prod_{u \in U} p_u}$.

  The integers $i_C$ for each $C \in \mathcal{C}$ and $s$ then form the input for \textsc{$(k+1)$-Product with Repetitions}
  
  If $C_1, C_2, \ldots, C_k \in \mathcal{C}$ is a solution of \probkEC
  then $s\prod_{i=1}^{k}i_{C_i} = 1$.
  
  Conversely, let $q_1, q_2, \ldots, q_{k+1}$ be a solution of the constructed \textsc{$(k+1)$-Product with Repetitions}.
  First of all, note that $s$ must be chosen precisely once.
  Indeed, all numbers except for $s$ are greater than 1 and thus, $s$ must be chosen at least once. If it were chosen more than once
  it would not be possible to cancel a power of $p^k$ in the denominator since the numerator would contain at most $p^{k-1}$.
  Therefore, the product of $q_1, q_2, \ldots, q_{k+1}$ is of the form $s i_{C_{j_k}} i_{C_{j_{k-1}}} \cdots i_{C_{j_1}} = 1$ which means
  that each prime representing an element of $U$ in the denominator is canceled. In other words, each element of $U$ is covered. Note also
  that since $s$ is chosen precisely once there cannot be any repetition within $i_{C_{j_k}} i_{C_{j_{k-1}}} \ldots i_{C_{j_1}}$.
  
  The reduction is parameterized since we only need the parameter $k$ for $k$ multiplications of $\frac{1}{p}$
  and first $n+1$ primes, where $n = |U|$, can be generated in time $\mathcal{O}(n^3)$ using, e.g., the Sieve of Eratosthenes for $(n+1)^2$.
  This follows from the fact, that the first $n$ primes lie among $1, \ldots, n^2$.
  For more details we refer the reader to Lemma~\ref{lem:n_th_prime} in Appendix (\ref{sec:appendix}).
\end{proof}
Let us now call the variant of \probVEST without $S$ and $\mathbf{v}$ \probMkProdRep. As we have mentioned in the introduction
we consider two cases regarding the target matrix. Namely, the Identity matrix and the Zero matrix:

\pbDef{\probMkProdZero}
{A list of $d \times d$ rational matrices and a parameter $k$.}
{Is it possible to choose $k$ matrices from the list (possibly with repetitions) such that their product is the $d\times d$ zero matrix?}

\pbDef{\probMkProdId}
{A list of $d \times d$ rational matrices and a parameter $k$.}
{Is it possible to choose $k$ matrices from the list (possibly with repetitions) such that their product is the $d\times d$ identity matrix?}
Note that \probMkProdId for $1\times 1$ matrices is exactly \probkProdRep. Therefore $\W[1]$-hardness for
\probMkProdId  for all matrix sizes follows from Theorem~\ref{thm:k_product_hard}.

Regarding \probMkProdZero, we can easily see that it is solvable in linear time for $1 \times 1$ matrices.
Indeed, it is sufficient to check whether $T_i=0$ for some $i$.
However, already for $2 \times 2$ matrices the problem becomes hard.
\begin{theorem}
\label{kProd0MatrHard}
 \probMkProdZero is $\W[1]$-hard for parameter $k$ even for $2 \times 2$ integer matrices.
\end{theorem}
\begin{proof}
  We reduce from \probAtmkSumRep. 
For every integer $x$ let us define 
  \begin{align*}
    U_x := 
    \renewcommand{\arraystretch}{1.2}
        \left(
          \begin{array}{ c c }
            1 & x \\
            0 & 1
          \end{array}
        \right).
  \end{align*}
It is easy to see that $U_xU_y = U_{x+y}$. 
Let $\mathcal{I}$ be an instance of \probAtmkSumRep with the set of integers $A$ 
and parameter $k$.
We create an equivalent instance $\mathcal{I}^\prime$ of \textsc{Matrix $(k+2)$-product with Repetitions} with the set of matrices
$\{U_a:a\in A\}\cup\{X\}$, where 
 \begin{align*}
    X= 
    \renewcommand{\arraystretch}{1.2}
        \left(
          \begin{array}{ c c }
            0 & 0 \\
           -1 & 1
          \end{array}
        \right).
  \end{align*}
 For correctness, assume that $\mathcal{I}$ is a YES-instance and $a_1, \ldots, a_\ell \in A$ are such that $\ell\le k$ and $\sum_{i=1}^\ell a_i=1$.
Consider the following product of $\ell+2$ matrices:
 \begin{align*}
    X\cdot \prod_{i=1}^l U_{a_i} \cdot X= X\cdot U_{\sum_{i=1}^\ell a_i} \cdot X = XU_1X=
    \renewcommand{\arraystretch}{1.2}
        \left(
          \begin{array}{ c c }
            0 & 0 \\
           -1 & 1
          \end{array}
        \right)
 \left(
          \begin{array}{ c c }
            1 & 1 \\
            0 & 1
          \end{array}
        \right)
 \left(
          \begin{array}{ c c }
            0 & 0 \\
           -1 & 1
          \end{array}
        \right)=
        \mathbf{0}.
  \end{align*}
For the other direction, assume that  $\mathcal{I}^\prime$ is a YES-instance.
Let $\ell$, $1\le \ell\le {k+2}$, be the minimal integer such that there are matrices $T_1,\dots, T_\ell$
from $\{U_a:a\in A\}\cup\{X\}$  with  $T_\ell T_{\ell-1} \cdots T_1 = \mathbf{0} \in \mathbb{Q}^{2 \times 2}$.
Since the matrix $X$ is idempotent (i.e. $X^2=X$), it does not appear two times in a row, otherwise we could reduce the length of the product.
Notice that $X$ should appear at least once, since the determinants of all $U_a$ are non-zero.
Assume that there is precisely one occurrence of $X$, then the product has form:
 \begin{align*}
    U_rXU_s=&
    \renewcommand{\arraystretch}{1.2}
        \left(
          \begin{array}{ c c }
            1 & r \\
            0 &1
          \end{array}
        \right)
 \left(
          \begin{array}{ c c }
            0 & 0 \\
           -1 & 1
          \end{array}
        \right)
 \left(
          \begin{array}{ c c }
            1 & s \\
            0 & 1
          \end{array}
        \right)=
 \left(
          \begin{array}{ c c }
            -r & -rs+r \\
            -1 & 1-s
          \end{array}
        \right)\neq \mathbf{0}.
  \end{align*} 
Hence, $X$ appears at least twice. Let us fix any two consequent occurrences and consider the partial product between them:
 \begin{align*}
    XU_rX=
    \renewcommand{\arraystretch}{1.2}
        \left(
          \begin{array}{ c c }
            0 & 0 \\
           -1 &1
          \end{array}
        \right)
 \left(
          \begin{array}{ c c }
            1 & r \\
           0 & 1
          \end{array}
        \right)
 \left(
          \begin{array}{ c c }
            0 & 0 \\
            -1 & 1
          \end{array}
        \right)=
 \left(
          \begin{array}{ c c }
            0 & 0 \\
            r-1 & 1-r
          \end{array}
        \right)=(1-r)\cdot X.
  \end{align*} 
If $r\ne 1$, we would get a shorter product resulting in zero, which contradicts to minimality of $\ell$. Hence $r=1$, so the product of $U_a$ that appear between two occurrences of $X$ is equal to $U_1$. Since there are at most $k$ of such $U_a$ and the sum of corresponding indices $a$ is equal to $1$, we obtain a solution to $\mathcal{I}$. 
\end{proof}

We can use similar approach to establish hardness of the \probVEST problem without $S$ (or alternatively when $S$ is the identity matrix).
Recall that here the task is to obtain not necessarily a zero matrix but any matrix which contains a given vector $\mathbf{v}$ in a kernel. 
\begin{theorem}\label{thm:k_product_2x2_with_v_hard}
 \probVEST is $\W[1]$-hard for parameter $k$ even for $2 \times 2$ integer matrices and when $S$ is the identity matrix.
\end{theorem}
The proof of this theorem is very similar to the proof of Theorem~\ref{kProd0MatrHard} and it can be found in Appendix (see~\ref{sec:appendix}).

At the end of this section, we show that \probVEST is equivalent to \probVEST without $S$ (in other words, when $S = I_d$)
and to \probMkProdZero.

\begin{theorem}\label{thm:without_S_equivalence}
  There is a parameterized reduction from \probVEST to the special case of \probVEST where $S$ is the identity matrix, and the other way around.
\end{theorem}
\begin{proof}
  One direction is trivial since the case when $S = I$ is just a special case of \probVEST.
   
  Regarding the other, let $\left(S \in \mathbb{Q}^{h \times d}, T_1, T_2, \ldots, T_m \in \mathbb{Q}^{d\times d}, \mathbf{v} \in \mathbb{Q}^d, k\right)$
  be an instance of \probVEST.
  First, we observe that without loss of generality we can suppose that $S$ is a square matrix (in other words, $h=d$).
  Indeed, if $h < d$ then we just add $d-h$ zero lines to $S$.  If $h > d$ we add $h-d$ zero columns to $S$,
		  $h-d$ zero entries to $\mathbf{v}$ and $h-d$ zero lines as well as $h-d$ zero columns to each $T_i$.
		  See Figure~\ref{f:square_S}.
		  		    
		  \begin{figure}
		  \begin{align*}
		  \renewcommand{\arraystretch}{1.2}
        \left(
        \begin{array}{ c c c c | c c c }
          \multicolumn{1}{|c}{} & & & & 0 & \hdots & \mc{0} \\
          \multicolumn{4}{|c|}{\raisebox{0\normalbaselineskip}[0pt][0pt]{$S$}} & \vdots & \ddots & \mc{\vdots} \\
          \multicolumn{1}{|c}{} & & & & 0 & \hdots & \mc{0}
        \end{array}
        \right) \in \mathbb{Q}^{h \times h}, 
        \left(
        \begin{array}{c c c | c c c }
          \multicolumn{1}{|c}{} & & & 0 & \hdots & 0 \\
          \multicolumn{3}{|c|}{\raisebox{0\normalbaselineskip}[0pt][0pt]{$T_i$}} & \vdots & \ddots & \vdots \\
          \multicolumn{1}{|c}{} & & & 0 & \hdots & 0 \\
          \cline{1-6}  
          0 & \hdots & 0 & 0 & \hdots & 0 \\
          \vdots & \ddots & \vdots & \vdots & \ddots & \vdots \\
          0 & \hdots & 0 & 0 & \hdots & 0
        \end{array}
        \right) \in \mathbb{Q}^{h \times h},
        \left(
        \begin{array}{ c }
          \mathbf{v} \\
          \cline{1-1}
          0 \\
          \vdots\\
          0
        \end{array}
        \right) \in \mathbb{Q}^{h}.
		  \end{align*}
		    \caption{A figure showing how to make all matrices square in the proof of Theorem~\ref{thm:without_S_equivalence}
		    when $h > d$.}\label{f:square_S}
		  \end{figure}
		  
    Now, we add 2 dimensions:
		To the vector $\mathbf{v}$ 
		we add $k$ on the $(d+1)$-st position and $1$ on the $(d+2)$-nd position.		
		To each matrix matrix $T_i$ we add a $2 \times 2$ submatrix which subtracts
		the $(d+2)$-nd component of a vector from the $(d+1)$-st.		
		To the matrix $S$ we add a submatrix which nullifies the $(d+2)$-nd component and multiplies the $(d+1)$-th component by 10.		
		Let $S^\prime, T^\prime_1, T^\prime_2 \ldots, T^\prime_m$
		denote the resulting $(d+2) \times (d+2)$ matrices and $\mathbf{v}^\prime$ denote the resulting $(d+2)$-dimensional vector.
		See Figure~\ref{f:without_S_equivalence}.
		The new parameter is set to $k+1$.
		
  	\begin{figure}
		  \begin{align*}
		    \renewcommand{\arraystretch}{1.2}
		    \mathbf{v}^\prime=
        \left(
        \begin{array}{ c }
          \mathbf{v} \\
          \cline{1-1}
          k \\
          1
        \end{array}
        \right),
		    S^\prime =
        \left(
        \begin{array}{c c c | c c }
          \multicolumn{1}{|c}{} & & & 0 & 0 \\
          \multicolumn{3}{|c|}{\raisebox{0\normalbaselineskip}[0pt][0pt]{$S$}} & \vdots & \vdots \\
          \multicolumn{1}{|c}{} & & & 0 & 0 \\
          \cline{1-5}  
          0 & \hdots & 0 & 10 & 0 \\
          0 & \hdots & 0 &  0 & 0
        \end{array}
        \right),
		    T^\prime_i=
        \left(
        \begin{array}{c c c | c c }
          \multicolumn{1}{|c}{} & & & 0 & 0 \\
          \multicolumn{3}{|c|}{\raisebox{0\normalbaselineskip}[0pt][0pt]{$T_i$}} & \vdots & \vdots \\
          \multicolumn{1}{|c}{} & & & 0 & 0 \\
          \cline{1-5}  
          0 & \hdots & 0 &  1 & -1 \\
          0 & \hdots & 0 &  0 &  1
        \end{array}
        \right).
		  \end{align*}
		    \caption{A construction forcing the matrix $S'$ to be selected last in the proof of Theorem~\ref{thm:without_S_equivalence}.}
		    \label{f:without_S_equivalence}
		  \end{figure}	
		  
		If there is a solution of the original problem, that is, there are $k$ matrices $T_{i_1}, \ldots, T_{i_k}$
		such that $ST_{i_k} T_{i_{k-1}} \cdots T_{i_1}\mathbf{v} = \mathbf{0}$,
		then $S^\prime T_{i_k}^\prime T_{i_{k-1}}^\prime \cdots T_{i_1}^\prime\mathbf{v}^\prime = \mathbf{0}$,
		since 1 is $k$ times subtracted from the $(d+1)$-st component of $\mathbf{v'}$ and the $(d+2)$-nd component is then nullified by $S^\prime$.
		
		Conversely, if there are $k+1$ matrices $Y_{1}, Y_{2}, \ldots, Y_{k+1}$,
		where each $Y_{i}$ is either $S^\prime$ or $T^\prime_{j}$
		for some $j$, such that $\mathbf{r} = Y_{k+1} Y_{k} \cdots Y_{1}\mathbf{v}^\prime =  \mathbf{0}$
		then $Y_{k+1}$ must be equal to $S^\prime$
		and the rest of the matrices are of type $T^\prime_j$, otherwise $\mathbf{r}_{d+1} \neq 0$ or $\mathbf{r}_{d+2} \neq 0$.
		Indeed, at first $k$ matrices of type $T^\prime_j$ must be selected to nullify the $(d+1)$-st component: if $Y_i=S^\prime$ for some $i\leq k$, this would increase the non-zero $(d+1)$-st component, so there would be no way to nullify it by remaining matrices $Y_{i+1}, \ldots, Y_{k+1}$. At the same time, $S^\prime$ should be necessarily selected once to nullify the $(d+2)$-nd component, so $Y_{k+1}=S^\prime$. Therefore, by restricting the matrices $Y_{1}, \ldots, Y_{k}$ to the first $d$ coordinates we obtain a solution to \probVEST with matrix $S$.
\end{proof}

\begin{theorem}\label{thm:equivalence_matrix_k_product_vest}
	\probVEST and \probMkProdZero are equivalent under parameterized reduction.
\end{theorem}
Note that one implication is relatively straightforward. Regarding the other, the idea is to again ``simulate'' the special matrix $S$ and the vector $\mathbf{v}$
by an ordinary matrix and force them to be selected as the leftmost and the rightmost, respectively. For the complete proof please see~\ref{sec:appendix}.

\section{Fixed-Parameter Tractability of \probVEST over Finite Fields}
While most of the hardness results for \probVEST and its variations in the previous section use constant-sized matrices, the entries of this matrices can be arbitrarily large. Here, we study the variation of the problem when all the matrices have entries from some finite field. Notice that restricting the field size by itself does not make the problem tractable:
recall the reduction from dominating set from Section~\ref{s:W[2]_hardness} which also works over $\Z_2$.
However, along with a bound on the matrix sizes this makes the problem tractable.
\begin{lemma}
Computing $M_k$ for a given instance of \probVEST over finite field $\F$ is $\FPT$ when parameterized by the size of $\F$ and the size of matrices.
\label{thm: fin_field_sizes}
\end{lemma}
\begin{proof}
Let $\mathcal M_{\F}^d$ be the set of all $d \times d$ matrices with entries from $\F$, then $|\mathcal M_{\F}^d|=|\F|^{d^2}$. For every $X \in \mathcal M_{\F}^d$ and every integer $i \in [k]$ we will compute a value $a_X^i \in \N_0$ equal to the number of sequences of $i$ matrices from the input such that their product is equal to $X$. In particular, this allows to obtain $M_k=\sum_{X\in \mathcal M_{\F}^d:\:SX\mathbf{v}=\mathbf{0}}a_X^k$.

For $i=1$ the computation can be done simply by traversing the input matrices. Assume that $a_X^i$ have been computed for all the matrices $X$ and all $i\in[j]$. We initiate by setting $a_X^{j+1}=0$ for every $X \in \mathcal M_{\F}^d$. Then, for every pair $(X,q)$, where $X \in \mathcal M_{\F}^d$ and $q\in[m]$, we increment $a_{XT_q}^{j+1}$ by $a_X^j$. In the end we will then have a correctly computed value $a_{Y}^{j+1}=\sum_{q=1}^{m}\sum_{X\::\:XT_q=Y} a_X^j$. 
\end{proof}

Our next step is to consider the matrices of unbounded size, but with at most $p$ first rows containing non-zero entries. In particular, if $\F=\Z_2$, we can associate to every such matrix $T$ a graph with the vertex set $[d]$ such that there exists an edge between the vertices $i$ and $j$, $i \le j$, if and only if $T^{i,j}=1$. Conversely, a graph with the vertex set $[d]$ can be represented by such a matrix if and only if the vertices in $[p]$ form it's vertex cover. 

Observe that every matrix $T$ with at most $p$ first non-zero rows has the following form:
        \begin{align*}
        T=\left(
        \begin{array}{c | c c c }
          A & & B &  \\
          \cline{1-4}  
          0 & 0 & \hdots & 0 \\
          \vdots & \vdots & \ddots & \vdots \\
          0 & 0 & \hdots & 0 \\
        \end{array}
        \right),  \text{ where $A$ is $p \times p$ matrix and $B$ is $p \times (d-p)$ matrix}.
        \end{align*}
 Further, we will denote matrices of this form by $A|B$.
 Consider the product of two such matrices $T_1=A_1|B_1$ and $T_2=A_2|B_2$:
 \begin{align*}
        \left(
        \begin{array}{c | c c c }
          A_1 & & B_1 &  \\
          \cline{1-4}  
          0 & 0 & \hdots & 0 \\
          \vdots & \vdots & \ddots  & \vdots \\
          0 & 0 & \hdots & 0 \\
        \end{array}
        \right)   
        \left(
        \begin{array}{c | c c c }
          A_2 & & B_2 &  \\
          \cline{1-4}  
          0 & 0 & \hdots & 0 \\
          \vdots & \vdots & \ddots  & \vdots \\
          0 & 0 & \hdots & 0 \\
        \end{array}
        \right)  
=\left(
        \begin{array}{c | c c c }
          A_1 A_2 & \multicolumn{3}{c}{A_1 B_2}  \\
          \cline{1-4}  
          0 & 0 & \hdots & 0 \\
          \vdots & \vdots & \ddots  & \vdots \\
          0 & 0 & \hdots & 0 \\
        \end{array}
        \right) = (A_1A_2)|(A_1B_2).
        \end{align*}
\begin{corollary}
\label{cor: decomp}
$\prod_{i=1}^k (A_i|B_i)= (\prod_{i=1}^k A_i)|(\prod_{i=1}^{k-1} A_i \cdot B_k)=(X A_k)|(X B_k) $, where\\ $X=\prod_{i=1}^{k-1}A_i$.
 In particular, the product does not depend on $B_i$ for $i<k$.
\end{corollary}
\thmFPT*
\begin{proof}
We slightly modify the definition of $a_X^i \in \N_0$ from the proof of Theorem~\ref{thm: fin_field_sizes}. Now, for every $i \in [k]$ and every matrix $X \in \mathcal M_{\F}^p$, let $a_X^i$ be the number of sequences of $i$ matrices $T_j=A_j|B_j$ from the input such that corresponding product of $A_j$ is equal to $X$.

The values of $a_X^i$ for every $i \in [k-1]$ can be computed same as in the proof of Theorem~\ref{thm: fin_field_sizes}. Given this information, we can count the sequences of length $k$ that nullify $\mathbf{v}$. Indeed, by Corollary \ref{cor: decomp}, the number of such sequences with the last matrix $T_j=A_j | B_j$ is precisely $b_j=\sum_{X \in \mathcal M_{\F}^p: \:S\cdot (XA_j)|(XB_j) \cdot \mathbf{v}=\mathbf{0}} a_X^{k-1}$. $M_k$ is then equal to $\sum_{j=1}^m b_j$.
\end{proof}

We remark that the algorithm for computing $a_X^i$ from the last proof can be exploited to determine minimal $k$ such that $M_k>0$, or to report that there is no such $k$. For this, let us run the algorithm with $k=1$, then with $k=2$ and so on. If after some iteration $k=j+1$ we obtain that $M_i=0$ for all $i \in [j]$ and there is no $X \in \mathcal M_{\F}^p$ such that $a_X^1=\dots=a_X^{j}=0$ and $a_X^{j+1}\neq 0$, we may conclude that $M_k=0$ for all $k \in N$, since every product of length more than $j$ can be obtained as a product of length at most $j$, and none of the latter nulify $\mathbf{v}$. Otherwise, there exists at least one $X\in \mathcal M_{\F}^p$ such that $a_X^1=\dots=a_X^{j}=0$ and $a_X^{j+1}\neq 0$. Note that every $X\in \mathcal M_{\F}^p$ can play this role only for one value of $k$. Therefore, it always suffices to make $|\mathcal M_{\F}^p|$ iterations of the algorithm. 

\section{Undecidability of \probVEST}\label{s:undecidability}
\newcommand{\probPost}{\textsc{Post's Correspondence Problem}\xspace}
In this section, we show that determining whether there exists $k \in \N$ such that $M_k > 0$ for an instance of \probVEST is an undecidable problem.
The reduction is from \probPost which is known to be undecidable. See \cite{Post47}.
\pbDef{(Binary) \probPost}
{$m$ pairs $(v_1,w_2), (v_2,w_2), \ldots, (v_m,w_m)$ of words over alphabet $\{0,1\}$.}
{Is possible to choose $k$ pairs $(v_{i_1}, w_{i_k}), (v_{i_2}, w_{i_2}) \ldots, (v_{i_k}, w_{i_k})$, for some $k \in \N$,
such that $v_{i_1}v_{i_2}\cdots v_{i_k} = w_{i_1}w_{i_2}\cdots w_{i_k}?$}

For a word $v \in \{0,1\}^\ast$ let $|v|$ be its length and let $(v)_2$ be the integer value of $v$ interpreting it as a binary number.
Let us define the following matrix for a binary word $v$.
\begin{align*}
T_v = 
  \begin{matrix}
    \begin{pmatrix}
    2^{|v|} - (v)_2 & (v)_2 \\
    2^{|v|} - (v)_2-1 & (v)_2+1 \\
  \end{pmatrix}
  \end{matrix}, \text{ then the following holds:}
\end{align*}
\begin{lemma}\label{lem:concatination_matrix}
	Let $v,w$ be binary words. Then, $T_v T_w = T_{wv}$ where $wv$ is the concatenation of $w$ and $v$.
\end{lemma}
Note that the construction of $T_v$ is a based on \cite{Claus71}[Satz 28, p. 157] which we are aware of thanks to G{\"u}nter Rote.
For the complete proof of Lemma~\ref{lem:concatination_matrix} please see \ref{sec:appendix}.

\paragraph*{Reduction}
Given an instance of \probPost we describe what an instance of \probVEST may look like.
For each pair $(v,w)$ we define

\begin{align*}
T_{(v,w)} = 
  \renewcommand{\arraystretch}{1.2}
  \left(
  \begin{array}{ c c | c c }
    \multicolumn{1}{|c}{} & & 0 & 0 \\
    \multicolumn{2}{|c|}{\raisebox{.6\normalbaselineskip}[0pt][0pt]{$T_{v}$}} & 0 & 0 \\
    \cline{1-4}
    0 & 0 & & \multicolumn{1}{c|}{}\\
    0 & 0 & \multicolumn{2}{c|}{\raisebox{.6\normalbaselineskip}[0pt][0pt]{$T_{w}$}} \\
  \end{array}
  \right),
\end{align*}
we set the initial vector $\mathbf{v} := \left(0,1,0,1\right)^T$ and $S :=\left(1,0,-1,0\right)$.
The undecidability of \probVEST then follows from the following lemma.
\begin{lemma}
	Let $(v_{i_1}, w_{i_1}), (v_{i_2}, w_{i_2}) \ldots, (v_{i_k}, w_{i_k})$ be $k$ pairs of binary words. Then 
	\begin{align*}
		S T_{(v_{i_k}, w_{i_k})}T_{(v_{i_{k-1}}, w_{i_{k-1}})} \cdots T_{(v_{i_1}, w_{i_1})} \mathbf{v} = \mathbf{0}	
	\end{align*}
	if and only if $v_{i_1}v_{i_2}\cdots v_{i_k} = w_{i_1}w_{i_2}\cdots w_{i_k}$.
\end{lemma}
\begin{proof}
	By Lemma~\ref{lem:concatination_matrix}	$T_{(v_{i_k}, w_{i_k})}T_{(v_{i_{k-1}}, w_{i_{k-1}})} \cdots T_{(v_{i_1}, w_{i_1})}
		= T_{(v_{i_1}v_{i_2}\cdots v_{i_k}, w_{i_1}w_{i_2}\cdots w_{i_k})}$.
The vector $\mathbf{v}$ selects the second column of the submatrix $T_{v_{i_1}v_{i_2}\cdots v_{i_k}}$
and the second column of the submatrix $T_{w_{i_1}w_{i_2}\cdots w_{i_k}}$. In other words,
the result is equal to
\begin{align*}
\left( \left( v_{i_1}v_{i_2}\cdots v_{i_k}\right)_2, \left( v_{i_1}v_{i_2}\cdots v_{i_k}\right)_2 + 1, \left( w_{i_1}w_{i_2}\cdots w_{i_k}\right)_2,
\left( w_{i_1}w_{i_2}\cdots w_{i_k}\right)_2 + 1     \right)^T.
\end{align*}
The final result after multiplying $S$ with the vector above is the following 1-dimensional vector
$\left( v_{i_1}v_{i_2}\cdots v_{i_k}\right)_2 - \left( w_{i_1}w_{i_2}\cdots w_{i_k}\right)_2$.
\end{proof}

\bibliographystyle{alpha}

\begin{thebibliography}{10}

\bibitem{AbboudLW14}
Amir Abboud, Kevin Lewi, and Ryan Williams.
\newblock Losing weight by gaining edges.
\newblock In Andreas~S. Schulz and Dorothea Wagner, editors, {\em Algorithms -
  {ESA} 2014 - 22th Annual European Symposium, Wroclaw, Poland, September 8-10,
  2014. Proceedings}, volume 8737 of {\em Lecture Notes in Computer Science},
  pages 1--12. Springer, 2014.

\bibitem{Anick85}
David~J. Anick.
\newblock {D}iophantine equations, {H}ilbert series, and undecidable spaces.
\newblock {\em Annals of {M}athematics}, 122:87--112, 1985.

\bibitem{Anick87}
David~J. Anick.
\newblock {G}eneric algebras and {C}{W} complexes.
\newblock {\em {A}lgebraic topology and algebraic {K}-theory}, pages 247--321,
  1987.

\bibitem{Anick89}
David~J. Anick.
\newblock The computation of rational homotopy groups is \#$\wp$-hard.
  {C}omputers in geometry and topology, {P}roc. {C}onf., {C}hicago/{I}ll. 1986,
  {L}ect. {N}otes {P}ure {A}ppl. {M}ath. 114.
\newblock pages 1--56, 1989.

\bibitem{Barkley41}
Rosser Barkley.
\newblock Explicit bounds for some functions of prime numbers.
\newblock {\em American Journal of Mathematics}, 63(1):211--232, 1941.

\bibitem{Boone59}
William~W. Boone.
\newblock The word problem.
\newblock {\em Annals of mathematics}, 70:207--265, 1959.

\bibitem{Brown57}
Edgar~H. Brown.
\newblock {F}inite computability of {P}ostnikov complexes.
\newblock {\em {A}nnals of {M}athematics}, 65:1, 1957.

\bibitem{Cadek14_2}
Martin {\v{C}}adek, Marek Kr{\v{c}}{\'a}l, Ji{\v{r}}{\'\i} Matou{\v{s}}ek,
  Luk{\'a}{\v{s}} Vok{\v{r}}{\'\i}nek, and Uli Wagner.
\newblock Extendability of continuous maps is undecidable.
\newblock {\em Discrete \& Computational Geometry}, 51(1):24--66, 2014.

\bibitem{Cadek14}
Martin {\v{C}}adek, Marek Kr\v{c}\'{a}l, Ji\v{r}\'{\i} Matou\v{s}ek,
  Luk\'{a}\v{s} Vok\v{r}\'{\i}nek, and Uli Wagner.
\newblock Polynomial-time computation of homotopy groups and {P}ostnikov
  systems in fixed dimension.
\newblock {\em {S}{I}{A}{M} {J}ournal on {C}omputing}, 43(5):1728--1780, 2014.

\bibitem{Claus71}
Volker Claus.
\newblock {\em Stochastische Automaten}.
\newblock Vieweg+Teubner Verlag, 1971.
\newblock (in German).

\bibitem{Downey95}
Rod~G. Downey and Michael~R. Fellows.
\newblock Fixed-parameter tractability and completeness {I}{I}: On completeness
  for {W}[1].
\newblock {\em Theoretical Computer Science}, 141(1):109--131, 1995.

\bibitem{flegg2001geometry}
Graham Flegg.
\newblock {\em From geometry to topology}.
\newblock Courier Corporation, 2001.

\bibitem{Flum04}
J\"{o}rg Flum and Martin Groge.
\newblock {\em Parameterized Complexity Theory}.
\newblock Springer, 2004.

\bibitem{Flum04_counting}
J\"{o}rg Flum and Martin Grohe.
\newblock The parameterized complexity of counting problems.
\newblock {\em SIAM Journal on Computing}, 33(4):892--922, 2004.

\bibitem{Matousek13}
Ji\v{r}\'{i} Matou\v{s}ek.
\newblock Computing higher homotopy groups is {W}[1]-hard.
\newblock {\em arXiv preprint arXiv:1304.7705}, 2013.

\bibitem{Novikov55}
Pyotr~S. Novikov.
\newblock On the algorithmic unsolvability of the word problem in group theory.
\newblock {\em Trudy {M}at. {I}nst. {S}teklov}, 44:1--143, 1955.
\newblock (in Russian).

\bibitem{Post47}
Emil~L. Post.
\newblock A variant of a recursively unsolvable problem.
\newblock {\em Journal of Symbolic Logic}, 12(2):55--56, 1947.

\bibitem{Roos79}
Jan-Erik Roos.
\newblock Relations between the poincar{\'e}-betti series of loop spaces and of
  local rings.
\newblock In {\em S{\'e}minaire d'Alg{\`e}bre Paul Dubreil}, pages 285--322.
  Springer, 1979.

\bibitem{whitehead1949combinatorial}
J.~H.~C. Whitehead.
\newblock Combinatorial homotopy. i.
\newblock {\em Bulletin of the American Mathematical Society}, 55(3):213--245,
  1949.

\bibitem{whitehead1949combinatorial2}
J.~H.~C. Whitehead.
\newblock {Combinatorial homotopy. II}.
\newblock {\em Bulletin of the American Mathematical Society}, 55(5):453 --
  496, 1949.

\end{thebibliography}

\appendix
\section{Appendix} \label{sec:appendix}
Our aim in the appendix is to give a complete proof of
\begin{itemize}
	\item $\W[1]$-hardnes of \probAtmkSumRep (Theorem~\ref{thm:at_most_k_sum_with_repetition_hardness}),
	\item Theorem~\ref{thm:k_product_2x2_with_v_hard},
	\item Theorem~\ref{thm:equivalence_matrix_k_product_vest},
	\item two auxiliary lemmas. Namely, Lemma~\ref{lem:n_th_prime} and Lemma~\ref{lem:concatination_matrix}.
\end{itemize}

\begin{theorem} \label{thm:at_most_k_sum_with_repetition_hardness}
    \probAtmkSumRep is $\W[1]$-hard when parameterized by $k$. 
    \label{l:ksum-rep-hard}
\end{theorem}
\begin{proof}
    Consider an instance $(U, \mathcal{C}, k)$ of \probUHS.
    Intuitively, we would like to model the sets in $\mathcal{C}$ as their characteristic vectors over $|U|$ dimensions,
    where each dimension corresponds to an element from $U$,
    and the vector representing a set $C \in \mathcal{C}$ is set to one exactly in those dimensions
    which correspond to the elements contained in the set which is represented by the vector.
    To model this in an instance of \probAtmkSumRep, we will represent said characteristic vectors as numbers in base $(k + 2)$.

    Formally, let $m = |U|$, $U = \{u_1, \ldots, u_{m}\}$, and $x = k + 2$. For each $C \in \mathcal{C}$,
    we add an element $a_C = -\left( x^{m+1} + \sum_{j ; u_j \in C} x^{j} \right)$ to the set $A$ of numbers.
    Then we also add to $A$ the number $y := kx^{m+1} + \sum_{j = 0}^{m} x^{j}$
    and we set the new parameter to $k+1$.
    Note that the numbers in $A$ are bounded by $x^{m+2}$,
    thus can be represented by $O(m \log k)$ bits, and $|A| = |\mathcal{C}|$,
    meaning that the reduction can be done in polynomial time.
    It remains to verify that the produced instance of \probAtmkSumRep is equivalent to the original instance of \probkEC.

    First, let $C_1, \ldots, C_k \in U$ be a solution to \probkEC.
    We claim that $\{y, a_{C_1}, \ldots, a_{C_k}\} \subset A$ is a solution to the instance $(A, k+1)$ of \probAtmkSumRep.
    Indeed, by construction and since each element of $U$ is covered exactly once, we have
    $a_{C_1} + \cdots + a_{C_k} =- \left(kx^{m+1} + \sum_{j=1}^m x^j\right) = -y+ x^0 = -y + 1$.
    Therefore, $y + a_{C_1} + \cdots + a_{C_k} = 1$.
    
    In the other direction, consider a solution $a_1, \dots, a_t \in A$ to \probAtmkSumRep where $t \leq k+1$.
    First of all, we observe that $y$ must be chosen precisely once.
    The sum $\sum_{j=1}^t a_t = 1 = x^0$ and $y$ is the only number with a coefficient (= 1) of $x^0$.
    Therefore, $y$ can be chosen $(\ell x + 1)$ times where $\ell \in \N_0$.
    However $t < x = k+2$. Whence, $\ell = 0$.
    In other words, $y$ is chosen precisely once and
    without loss of generality, we suppose that $a_1 = y$.
    
    Next, we show that $t=k+1$.
    The number $y$ which is chosen precisely once has $k$ as the coefficient of $x^{m+1}$ which has to be nullified.
    The only option how to do that is to choose $k$ numbers other than $y$.
    (Such numbers are negative and have $1$ as a coefficient of $x^{m+1}$.)
        
    Finally, from the equality $\sum_{j=2}^{k+1} a_j = - y + 1 = -kx^{m+1} - \sum_{j=1}^m x^j$ we conclude
    that no $-x^i$ for $i \leq m$ is contained in more than one $a_j$ as a summand since $k< k+2 = x$.
    By the same argument we observe that each $-x^i$ is contained in some $a_j$ as a summand.
    Indeed, addition of at most $k$ terms $-x^i$ cannot affect coefficient of $x^{i+1}$.
    Therefore, each $-x^i$ for $i \leq m$ is contained in precisely one $a_j$ and thus,
    $\{C; a_C \in \{a_2, \ldots, a_{k+1}\}\}$ is a desired $k$-exact cover.   
\end{proof}

\begin{proof}[Proof of Theorem~\ref{thm:k_product_2x2_with_v_hard}]
As in the proof of Theorem~\ref{kProd0MatrHard}, we proceed by reduction from \probAtmkSumRep. 
Let $\mathcal{I}$ be an arbitrary instance of the problem with the set of integers $A$ 
and parameter $k$.
We create an equivalent instance $\mathcal{I}^\prime$ of \probVEST with parameter $k+1$, vector $v=(0,1)^T$
and the set of matrices $\{U_a:a\in A\}\cup\{X\}$, where $U_a$ and $X$ are defined same as in the proof of Theorem \ref{kProd0MatrHard}.
We set $S$ equal to the identity matrix. 

For correctness, assume that $\mathcal{I}$ is a YES-instance and $a_1, \ldots, a_\ell \in A$ are such that $\ell\le k$ and $\sum_{i=1}^\ell a_i=1$.
We apply the following $\ell+1$ matrices to nullify $v$:
 \begin{align*}
    X\cdot \prod_{i=1}^\ell U_{a_i} \cdot v= XU_1v=
    \renewcommand{\arraystretch}{1.2}
        \left(
          \begin{array}{ c c }
            0 & 0 \\
           -1 & 1
          \end{array}
        \right)
 \left(
          \begin{array}{ c c }
            1 & 1 \\
            0 & 1
          \end{array}
        \right)
 \left(
          \begin{array}{ c c }
            0 \\
           1
          \end{array}
        \right)=
\left(
          \begin{array}{ c c }
           0  \\
           0
          \end{array}
        \right).
  \end{align*}
For another direction, assume that  $\mathcal{I}^\prime$ is a YES-instance.
Let $\ell$, $1\le \ell\le {k+1}$, be the minimal integer such that $T_\ell \cdots T_1v = (0,0)^T$
for some $T_1,\ldots, T_\ell$ from $\{U_a:a\in A\}\cup\{X\}$.
Since the determinants of all $U_a$ are non-zero, $T_i=X$ for some $i\in[\ell]$.
Observe that $Xv=v$, so by minimality of $\ell$ we have that $T_1\ne X$.
Let $i$ be the minimal index such that $T_i=X$, $2\le i \le \ell$.
Then $T_{i-1} \cdots T_1=U_s$ for some integer $s$. Let us apply first $i$ matrices to $v$:
 \begin{align*}
    T_i \cdots T_1v = XU_s\cdot v &=
    \renewcommand{\arraystretch}{1.2}
 \left(
          \begin{array}{ c c }
            0 & 0 \\
           -1 & 1
          \end{array}
        \right)
 \left(
          \begin{array}{ c c }
            1 & s \\
            0 & 1
          \end{array}
        \right)
\left(
          \begin{array}{ c c }
           0 \\
           1 
          \end{array}
        \right)
        \\&=
 \left(
          \begin{array}{ c c }
            0 & 0 \\
            -1 & 1
          \end{array}
        \right)
\left(
          \begin{array}{ c c }
           s \\
           1 
          \end{array}
        \right)=
\left(
          \begin{array}{ c c }
            0 \\
            1-s
          \end{array}
        \right).
  \end{align*} 
If $s\ne 1$, we get a multiple of $v$, which is in contradiction to minimality of $l$.
So $T_{i-1},\ldots,T_1=U_1$, which is a product of at most $k$ matrices of the form $U_a$ with $a\in A$.
The sum of corresponding indices $a$ is then equal to $1$, resulting in a solution to $\mathcal{I}$.
\end{proof}

\begin{proof}[Proof of Theorem~\ref{thm:equivalence_matrix_k_product_vest}]\hfill
	\begin{enumerate}
		\item ``Parameterized reduction from \probMkProdZero to \probVEST''\newline
		For each matrix $T_i \in \mathbb{Q}^{d \times d}$ 
		we introduce a block matrix $T^\prime_i \in \mathbb{Q}^{d^2 \times d^2}$ whose each block is $T_i$.
		We set $\mathbf{v} = (e_1, e_2, \ldots, e_d)^T \in \mathbb{Q}^{d^2}$ where each $e_i$ is the $d$-dimensional
		unit vector with 1 on its $i$-th coordinate and $S$ to the $d^2$-dimensional identity matrix.
    Therefore, $T_{i_k}T_{i_{k-1}}\cdots T_{i_1} = R$ if and only if
    $ST^\prime_{i_k} \cdots T^\prime_{i_1}\mathbf{v} = (R_{\ast,1}, R_{\ast,2}, \ldots, R_{\ast,d})^T$
    where $R_{\ast, j}$ is the $j$-th column of the matrix $R$.		
		
		\item ``Parameterized reduction from \probVEST to \probMkProdZero''\newline
		We first reduce \probVEST to the version of \probVEST without $S$ as we did in the proof of Theorem~\ref{thm:without_S_equivalence}.
		Thus, we assume that our input consists of the initial vector $\mathbf{v}$, square matrices $S^\prime, T^\prime_{1}, \ldots, T^\prime_{m}$,
		where $S^\prime$ represents the original special matrix $S$, and the parameter is $k+1$.
		Let us recall that $S^\prime$ has to be selected precisely once as the leftmost matrix otherwise the resulting vector cannot be zero
		by the construction from the proof of Theorem~\ref{thm:without_S_equivalence}.
		
		Now, we create an instance of \textsc{Matrix $(k+3)$-Product with Repetitions}.
		Let $T_\mathbf{v}$ be a matrix containing the vector $\mathbf{v}$ in the first column and zero otherwise.
		The idea is to use the matrix $T_\mathbf{v}$ instead of the vector $\mathbf{v}$ and force such matrix
		to be selected as the rightmost after $S^\prime$ and $k$ matrices of type $T^\prime_{i}$ by adding some blocks.		
    We use the construction from the proof of Theorem~\ref{kProd0MatrHard}. Namely, we use matrices $X$ and
    $U_{-2}$ and $U_{2k+1}$ as submatrices. By the same argument as in the proof of Theorem~\ref{kProd0MatrHard}
    the only way how to make the zero matrix by multiplying $k+3$ matrices from $\{X, U_{-2}, U_{2k+1}\}$ is to choose $X$ twice,
    as the leftmost and the rightmost matrix, $k$-times $U_{-2}$ and once $U_{2k+1}$ as intermediate matrices.
    Therefore, we can add $X$ to $T_\mathbf{v}$ and to the identity matrix as block submatrices,
    $U_{2k+1}$ to $S^\prime$ (since $S^\prime$ must be selected precisely once) and $U_{-2}$ to $T_i^\prime$.
    It remains to force the order of $T_\mathbf{v}$ and the identity matrix enriched by $X$. For this, we add
    submatrices $A, B$ such that $AB = 0$ while $BA \neq 0, AA \neq 0, BB \neq 0$. We add $A$ to the identity matrix
    enriched by $X$, $B$ to $T_\mathbf{v}$ enriched by $X$ and identity matrices to the rest. See Figure~\ref{f:without_S_v_equivalence}. 
    The following settings for $A$ and $B$, respectively, work.
    		\begin{align*}
		  A =
		  \renewcommand{\arraystretch}{1.2} 
        \left(
        \begin{array}{ c c c }
          0 & 0 & 0 \\
          0 & 1 & 0 \\
          0 & 0 & 0
        \end{array}
        \right),
        B =
		    \renewcommand{\arraystretch}{1.2} 
        \left(
        \begin{array}{ c c c }
          0 & 1 & 0 \\
          0 & 0 & 0 \\
          0 & 0 & 1
        \end{array}
        \right).
		\end{align*}
    \begin{figure}
		  \begin{align*}
        T_i^{\prime\prime} &=
        \renewcommand{\arraystretch}{1.2} 
        \left(
        \begin{array}{ c c c | c c c | c c }
          \multicolumn{1}{|c}{} & & & 0 & 0 & \mc{0} & 0 & 0 \\
          \multicolumn{3}{|c|}{\raisebox{0\normalbaselineskip}[0pt][0pt]{$T^\prime_i$}} & \vdots & \vdots & \mc{\vdots} & \vdots & \vdots \\
          \multicolumn{1}{|c}{} & & & 0 & 0 & \mc{0} & 0 & 0 \\
          \cline{1-6}
          0 & \hdots & 0 & \multicolumn{1}{c}{} & & & 0 & 0 \\
          0 & \hdots & 0 & \multicolumn{3}{c|}{\raisebox{0\normalbaselineskip}[0pt][0pt]{$I_3$}} & 0 & 0 \\
          0 & \hdots & 0 & \multicolumn{1}{c}{} & & & 0 & 0 \\
          \cline{4-8}
          0 & \hdots & \mc{0} & 0 & 0 & 0 & & \multicolumn{1}{c|}{}   \\
          0 & \hdots & \mc{0} & 0 & 0 & 0 & \multicolumn{2}{c|}{\raisebox{.6\normalbaselineskip}[0pt][0pt]{$U_{-2}$}}
        \end{array}
        \right),
        S^{\prime\prime} =
        \left(
        \begin{array}{ c c c | c c c | c c }
          \multicolumn{1}{|c}{} & & & 0 & 0 & \mc{0} & 0 & 0 \\
          \multicolumn{3}{|c|}{\raisebox{0\normalbaselineskip}[0pt][0pt]{$S^\prime$}} & \vdots & \vdots & \mc{\vdots} & \vdots & \vdots \\
          \multicolumn{1}{|c}{} & & & 0 & 0 & \mc{0} & 0 & 0 \\
          \cline{1-6}
          0 & \hdots & 0 & \multicolumn{1}{c}{} & & & 0 & 0 \\
          0 & \hdots & 0 & \multicolumn{3}{c|}{\raisebox{0\normalbaselineskip}[0pt][0pt]{$I_3$}} & 0 & 0 \\
          0 & \hdots & 0 & \multicolumn{1}{c}{} & & & 0 & 0 \\
          \cline{4-8}
          0 & \hdots & \mc{0} & 0 & 0 & 0 & & \multicolumn{1}{c|}{}   \\
          0 & \hdots & \mc{0} & 0 & 0 & 0 & \multicolumn{2}{c|}{\raisebox{.6\normalbaselineskip}[0pt][0pt]{$U_{2k+1}$}}
        \end{array}
        \right),
        \\
        T_\mathbf{v}^\prime &=
        \renewcommand{\arraystretch}{1.2} 
        \left(
        \begin{array}{ c c c | c c c | c c }
          \multicolumn{1}{|c}{} & & & 0 & 0 & \mc{0} & 0 & 0 \\
          \multicolumn{3}{|c|}{\raisebox{0\normalbaselineskip}[0pt][0pt]{$T_\mathbf{v}$}} & \vdots & \vdots & \mc{\vdots} & \vdots & \vdots \\
          \multicolumn{1}{|c}{} & & & 0 & 0 & \mc{0} & 0 & 0 \\
          \cline{1-6}
          0 & \hdots & 0 & \multicolumn{1}{c}{} & & & 0 & 0 \\
          0 & \hdots & 0 & \multicolumn{3}{c|}{\raisebox{0\normalbaselineskip}[0pt][0pt]{$B$}} & 0 & 0 \\
          0 & \hdots & 0 & \multicolumn{1}{c}{} & & & 0 & 0 \\
          \cline{4-8}
          0 & \hdots & \mc{0} & 0 & 0 & 0 & & \multicolumn{1}{c|}{}   \\
          0 & \hdots & \mc{0} & 0 & 0 & 0 & \multicolumn{2}{c|}{\raisebox{.6\normalbaselineskip}[0pt][0pt]{$X$}}
        \end{array}
        \right),
        H =
        \left(
        \begin{array}{ c c c | c c c | c c }
          \multicolumn{1}{|c}{} & & & 0 & 0 & \mc{0} & 0 & 0 \\
          \multicolumn{3}{|c|}{\raisebox{0\normalbaselineskip}[0pt][0pt]{$I_d$}} & \vdots & \vdots & \mc{\vdots} & \vdots & \vdots \\
          \multicolumn{1}{|c}{} & & & 0 & 0 & \mc{0} & 0 & 0 \\
          \cline{1-6}
          0 & \hdots & 0 & \multicolumn{1}{c}{} & & & 0 & 0 \\
          0 & \hdots & 0 & \multicolumn{3}{c|}{\raisebox{0\normalbaselineskip}[0pt][0pt]{$A$}} & 0 & 0 \\
          0 & \hdots & 0 & \multicolumn{1}{c}{} & & & 0 & 0 \\
          \cline{4-8}
          0 & \hdots & \mc{0} & 0 & 0 & 0 & & \multicolumn{1}{c|}{}   \\
          0 & \hdots & \mc{0} & 0 & 0 & 0 & \multicolumn{2}{c|}{\raisebox{.6\normalbaselineskip}[0pt][0pt]{$X$}}
        \end{array}
        \right).
      \end{align*}
      \caption{The instance of \probMkProdZero obtained after the reduction from \probVEST
      	in the proof of Theorem~\ref{thm:equivalence_matrix_k_product_vest}.}\label{f:without_S_v_equivalence}
    \end{figure}	
		
	\end{enumerate}
\end{proof}

\begin{lemma}\label{lem:n_th_prime}
  Let $p_n$ denote the $n$-th prime. Then $p_n \leq n^2$ for $n \geq 2$.
\end{lemma}
\begin{proof}
Let $\pi(x)$ denote the number of primes less than or equal to $x$. The lemma follows, e.g., from the following claims:
\begin{itemize}
  \item $p_n < n \left( \ln n + \ln \ln n \right)$ for $6 \leq n \leq e^{95}$ (see~\cite[Theorem~28]{Barkley41}),
  \item $\frac{x}{\ln x + 2} \leq \pi(x)$ for $x \geq 55$ (see~\cite[Theorem~29.A]{Barkley41}),
  \item $p_2 = 3$, $p_3 = 5$, $p_4 = 7$, $p_5 = 11$.
\end{itemize}
\end{proof}

\begin{proof}[Proof of Lemma~\ref{lem:concatination_matrix}.]
	First of all, note that $2^{|v|}(w)_2 + (v)_2 = (wv)_2$. Using this observation, we compute all the entries
	of the matrix
	\begin{align*}
		 T_v T_w = 
		 \begin{matrix}
    \begin{pmatrix}
    2^{|v|} - (v)_2 & (v)_2 \\
    2^{|v|} - (v)_2-1 & (v)_2+1 \\
  \end{pmatrix}
  \end{matrix}
  \begin{matrix}
    \begin{pmatrix}
    2^{|w|} - (w)_2 & (w)_2 \\
    2^{|w|} - (w)_2-1 & (w)_2+1 \\
  \end{pmatrix}
  \end{matrix}.
  \end{align*}
  Thus,
  \begin{align*}
		\left(T_{v}T_{w}\right)^{1,1}=& \left( 2^{|v|} - (v)_2 \right) \left( 2^{|w|} - (w)_2 \right) + (v)_2 \left( 2^{|w|} - (v)_2-1  \right)\\
		=& 2^{|wv|} - 2^{|v|}(w)_2 - 2^{|w|}(v)_2 + (v)_2(w)_2 + 2^{|w|}(v)_2 - (v)_2(w)_2 - (v)_2\\
		=& 2^{|wv|} - 2^{|v|}(w)_2 - (v)_2 \\
		=& 2^{|wv|} - (wv)_2,\\
		\left(T_{v}T_{w}\right)^{1,2} =& \left( 2^{|v|} - (v)_2 \right) (w)_2 +  (v)_2 \left((w)_2+1\right)\\
		=& 2^{|v|}(w)_2 - (v)_2 (w)_2 + (v)_2 (w)_2 + (v)_2\\
		=& 2^{|v|}(w)_2 + (v)_2\\
		=& (wv)_2,\\
		\left(T_{v}T_{w}\right)^{2,1} =& \left( 2^{|v|} - (v)_2-1 \right) \left( 2^{|w|} - (w)_2 \right) + \left( (v)_2+1 \right) \left(  2^{|w|} - (w)_2-1 \right)\\
		=& 2^{|wv|} - 2^{|v|}(w)_2 - 2^{|w|}(v)_2 +(v)_2(w)_2 - 2^{|w|} + (w)_2 \\
		&+ 2^{|w|}(v)_2 - (v)_2 (w)_2 - (v)_2 +2^{|w|} - (w)_2 - 1\\
		=& 2^{|wv|}  - 2^{|v|}(w)_2 - (v)_2 - 1 \\
		=& 2^{|wv|} - (wv)_2 - 1,\\
		\left(T_{v}T_{w}\right)^{2,2} =& \left( 2^{|v|} - (v)_2-1 \right) (w)_2 + \left( (v)_2+1 \right) \left( (w)_2+1 \right)\\
		=& 2^{|v|}(w)_2 - (v)_2(w)_2 - (w)_2 + (v)_2(w)_2 + (v)_2 + (w)_2 + 1\\
		=& 2^{|v|}(w)_2 + (v)_2 + 1\\
		=& (wv)_2 + 1.
	\end{align*}
\end{proof}

\end{document}